\makeatletter\let\@twosidetrue\@twosidefalse\let\@mparswitchtrue\@mparswitchfalse\makeatother %
\documentclass[a4paper,runningheads,dvipsnames]{llncs}
\usepackage[T1]{fontenc}

\usepackage{etoolbox}
\newtoggle{arxiv}
\toggletrue{arxiv} %
\iftoggle{arxiv}{
\pagestyle{headings}
\usepackage[paperheight=235mm, paperwidth=155mm,textwidth=12.2cm,textheight=19.3cm]{geometry}
}{}
\newcommand{\ifarxivelse}[2]{\iftoggle{arxiv}{#1}{\cite[#2]{arxivversion}}}

\usepackage{amsmath,amssymb}
\usepackage{graphicx}
\usepackage{enumitem}
\usepackage{cite}
\usepackage{bbding}
\usepackage{xspace}
\usepackage{csquotes}
\usepackage{todonotes}
\usepackage{bbm}
\usepackage{xcolor}
\usepackage{booktabs}
\usepackage[%
  bookmarks,
  unicode,
  colorlinks=true,
  allcolors=blue!65!black!90,
  breaklinks=true]{hyperref}

\usepackage{tikz}
\usetikzlibrary{positioning,automata,shapes.geometric,decorations,decorations.footprints,calc,fit,arrows.meta,backgrounds}  %
\usepackage{pgfplots}\pgfplotsset{compat=1.17}

\tikzstyle{state}+=[minimum size = 6mm, inner sep=0,outer sep=1]

\colorlet{disabled}{lightgray}
\tikzstyle{state}=[draw,rectangle,inner sep=5pt,rounded corners=2pt]
\tikzstyle{action}=[font=\small,inner sep=0pt,outer sep=3pt]
\tikzstyle{actionnode}=[circle,draw=black,fill=black,minimum size=1mm,inner sep=0,outer sep=0]
\tikzstyle{actionedge}=[draw,-]
\tikzstyle{prob}=[font=\scriptsize,inner sep=0pt,outer sep=1pt]
\tikzstyle{probedge}=[draw,->]
\tikzstyle{directedge}=[draw,->]
\tikzset{chainarrow/.tip={Stealth[length=3pt]}}
\tikzset{>=chainarrow}

\tikzstyle{discontinuity-limit}=[circle,inner sep=0pt,minimum size=3pt,fill=white,draw=black]
\tikzstyle{discontinuity-value}=[circle,inner sep=0pt,minimum size=3pt,fill=black]
\tikzstyle{discontinuity-limit-small}=[circle,inner sep=0pt,minimum size=1pt,fill=white,draw=black,solid]
\tikzstyle{discontinuity-value-small}=[circle,inner sep=0pt,minimum size=1pt,fill=black,solid]

\setitemize{noitemsep,topsep=0pt,parsep=0pt,partopsep=0pt,leftmargin=12.0pt}
\setenumerate{noitemsep,topsep=0pt,parsep=0pt,partopsep=0pt}
\setdescription{noitemsep,topsep=2pt,parsep=0pt,partopsep=0pt,leftmargin=12.5pt}

\usepackage[nameinlink,capitalize,english]{cleveref}
\Crefname{figure}{Fig.}{Figs.}
\crefname{figure}{Fig.}{Figs.}
\Crefname{tabular}{Tab.}{Tabs.}
\crefname{tabular}{Tab.}{Tabs.}
\Crefname{section}{Sec.}{Sects.}
\crefname{section}{Sec.}{Sects.}
\Crefname{appendix}{App.}{Apps.}
\crefname{appendix}{App.}{Apps.}
\Crefname{equation}{Eq.}{Eqs.}
\crefname{equation}{Eq.}{Eqs.}
\creflabelformat{equation}{#2#1#3}
\Crefname{example}{Ex.}{Exs.}
\crefname{example}{Ex.}{Exs.}

\newcommand{\RR}{\ensuremath{\mathbb{R}}\xspace}  %
\newcommand{\RRpos}{\ensuremath{\RR_{\geqslant0}}\xspace}  %

\newcommand{\tool}[1]{\textsc{#1}}
\newcommand{\lang}[1]{\textsc{#1}}

\newcommand{\toolset}{\tool{Modest Toolset}\xspace}
\newcommand{\eg}{e.g.\ }
\newcommand{\ie}{i.e.\ }

\newcommand{\wrt}{w.r.t.\xspace}

\newcommand{\set}[1]{\ensuremath{\{\,#1\,\}}}
\newcommand{\tuple}[1]{\ensuremath{\langle #1 \rangle}}

\newcommand{\defeq}{\mathrel{\vbox{\offinterlineskip\ialign{\hfil##\hfil\cr{\tiny \rm def}\cr\noalign{\kern0.30ex}$=$\cr}}}}

\DeclareMathOperator{\Dist}{\mathcal{D}}

\DeclareMathOperator{\CVaR}{CVaR}

\DeclareMathOperator{\ERisk}{ERisk}
\DeclareMathOperator{\Quant}{Q}
\newcommand{\stochdom}{\precsim_\mathit{SD}}
\newcommand{\domain}{\mathit{dom}}
\newcommand{\pz}{\phantom{0}}

\makeatletter%
\g@addto@macro\normalsize{%
  \setlength\abovedisplayskip{2pt}%
  \setlength\belowdisplayskip{2pt}%
  \setlength\abovedisplayshortskip{-2pt}%
  \setlength\belowdisplayshortskip{2pt}%
}%
\makeatother

\makeatletter
\def\orcidID#1{\textsuperscript{\,\smash{\protect\raisebox{-1.25pt}{\href{http://orcid.org/#1}{\protect\includegraphics[scale=.8]{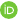}}}}}}
\makeatother

\begin{document}
\title{%
Statistical Model Checking Beyond Means:
Quantiles, CVaR, and the DKW Inequality%
\thanks{
This work was supported
by the DFG through the Cluster of Excellence EXC 2050/1 (CeTI, project ID 390696704, as part of Germany's Excellence Strategy)
and the TRR 248 (see \url{perspicuous-computing.science}, project ID 389792660),
by the ERC Starting Grant DEUCE (101077178),
by the European Union's Horizon 2020 research and innovation programme under Marie Sk{\l}odowska-Curie grant agreement 101008233 (MISSION),
by the Interreg North Sea project STORM\_SAFE,
and
by NWO VIDI grant VI.Vidi.223.110 (TruSTy).
}%
\iftoggle{arxiv}{\\ (extended version)}{}%
}
\titlerunning{SMC Beyond Means: Quantiles, CVaR, and the DKW Inequality}
\author{
\mbox{Carlos E.\ Budde\inst{1}\orcidID{0000-0001-8807-1548}}
\and
\mbox{Arnd Hartmanns\inst{2}\orcidID{0000-0003-3268-8674}}
\and
\mbox{Tobias Meggendorfer\inst{3}\orcidID{0000-0002-1712-2165}}
\and
\mbox{Maximilian Weininger\inst{4}\orcidID{0000-0002-0163-2152}}
\and
Patrick Wienh\"oft\inst{5,6}\orcidID{0000-0001-8047-4094}%
}
\authorrunning{C.E. Budde, A. Hartmanns, T. Meggendorfer, M. Weininger, P. Wienhöft}
\institute{%
Technical University of Denmark, Lyngby, Denmark
\and
University of Twente, Enschede, The Netherlands
\and
Lancaster University Leipzig, Leipzig, Germany
\and
Ruhr University Bochum, Bochum, Germany
\and
Dresden University of Technology, Germany
\and
Centre for Tactile Internet with Human-in-the-Loop (CeTI), Dresden, Germany
}

\maketitle

\begin{abstract}
Statistical model checking (SMC) randomly samples probabilistic models to approximate quantities of interest with statistical error guarantees.
It is traditionally used to estimate probabilities and expected rewards, \ie means of different random variables on paths.
In this paper, we develop methods using the Dvoretzky-Kiefer-Wolfowitz-Massart inequality (DKW) to extend SMC beyond means to compute quantities such as quantiles, conditional value-at-risk, and entropic risk.
The DKW provides confidence bounds on the random variable's entire cumulative distribution function, a much more versatile guarantee compared to the statistical methods prevalent in SMC today.
We have implemented support for computing new quantities via the DKW in the \tool{modes} simulator of the \toolset.
We highlight the implementation and its versatility on benchmarks from the quantitative verification literature.
\end{abstract}

\section{Introduction}

Statistical model checking (SMC)~\cite{AP18,YS02,Kre16,LLTYSG19} avoids the state space explosion problem of classic probabilistic model checking approaches (PMC)~\cite{BAFK18,Bai16} that explore and numerically analyse a model's entire state space~\cite{HJQW23}:
SMC instead \emph{samples} $k$ random paths from the model to \emph{estimate} the value of the quantity of interest.
As a simulation-based approach, it applies to any effectively executable model, including non-Markovian~\cite{DM18} and hybrid~\cite{EGF15,PR17} ones.
An SMC result comes with a \emph{statistical} correctness guarantee, often expressed as a \emph{confidence interval} $[l, u]$ that contains the true result $(1-\delta) \cdot 100\,\%$ of the times~\cite{BHMWW25}.

The most fundamental quantities computed by PMC and SMC are reachability probabilities and expected rewards~\cite{BHMWW25,HJQW23}.
SMC estimates these quantities using statistical methods like the Clopper-Pearson confidence interval~\cite{CP34} for probabilities (\ie binomial proportions) and Hoeffding's inequality~\cite{Hoe63} for means of bounded distributions, or compares them to each other~\cite{DLLMW11} or to thresholds using Wald's sequential probability ratio test~\cite{Wal45}.
In the past decade, PMC has been extended to compute several other quantities of interest, such as quantiles/percentiles/value-at-risk~\cite{UB13,KSBD15,RRS17}, conditional value at risk~\cite{KM18}, and entropic risk~\cite{BCMP24}.
However, the application of SMC has so far been limited to probabilities and expected rewards, \ie only the \emph{means} of distributions associated to different random variables on sampled paths.
To the best of our knowledge, no SMC approaches or tools support quantities other than means yet.

In this paper, we show how to extend SMC to estimate non-mean quantities using the Dvoretzky-Kiefer-Wolfowitz-Massart inequality (DKW)~\cite{DKW56,Mas90}.
The DKW provides a sound \emph{simultaneous confidence band} around the cumulative distribution function (cdf), \ie upper and lower bound functions completely enveloping the (unknown) cdf $(1-\delta) \cdot 100\,\%$ of the times %
(see \Cref{fig:example}).
This is a stronger statement compared to the currently-used statistical methods for estimation mentioned above, as it applies to the entire cdf rather than a single point or pointwise.
From the DKW, we can again derive a confidence interval for the mean~\cite{BHMWW25,And69}, but equally (and simultaneously) obtain confidence intervals on other quantities as well.
We show how to do so in particular for higher moments, quantiles, conditional value-at-risk, and entropic risk.
We have implemented these DKW-based computations in the \tool{modes} statistical model checker~\cite{BDHS20}, part of the \toolset\cite{HH14}.
\tool{modes} can now also export the empirical cdf and DKW confidence band for plotting and further analysis by the user.
We highlight our implementation and its versatility using several models from the Quantitative Verification Benchmark Set (QVBS)~\cite{HKPQR19} in \Cref{sec:Implementation}.

\section{Preliminaries}\label{sec:prelims}

A \emph{probability distribution} over a non-empty, countable set $S$ is a function $\mu\colon S \to [0, 1]$ such that $\sum_{s \in S} \mu(s) = 1$.
The set of all distributions over $S$ is denoted by $\Dist(S)$.
The \emph{cumulative distribution function} (cdf) of a random variable $X$ is given by $F_X(x) \defeq \mathbb{P}(X \leq x)$.
A random variable $X$ \emph{stochastically dominates} another random variable $Y$, written $Y \stochdom X$, if $F_Y(x) \geq F_X(x)$ for all $x$ (\ie for any $x$, obtaining a value less than or equal to $x$ is more likely for $Y$ than for $X$; intuitively, $X$ yields larger values).
If $Y \stochdom X$, then $\mathbb{E}(Y) \leq \mathbb{E}(X)$.

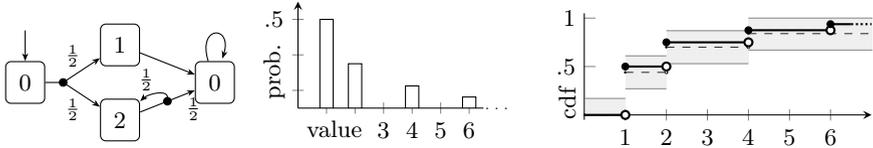
\begin{figure}[t]
    \centering%
\begin{minipage}{.275\textwidth}
\centering
\begin{tikzpicture}[auto,xscale=1.25]
    \node[state,initial above,initial text=] at (0,0) (init) {$0$};
    \node[state] at (1,.5) (s1) {$1$};
    \node[state] at (1,-.5) (s2) {$2$};
    \node[state] at (2,0) (sf) {$0$};
    \node[actionnode] at (.4,0) (inita) {};
    \node[actionnode] at (1.5,-.25) (s2a) {};

    \path[actionedge]
        (init) edge (inita)
        (s2) edge (s2a)
    ;
    \path[directedge]
        (s1) edge (sf)
        (sf) edge[loop above] (sf)
    ;
    \path[probedge]
        (inita) edge node[prob] {$\tfrac{1}{2}$} (s1)
        (inita) edge[swap] node[prob] {$\tfrac{1}{2}$} (s2)
        (s2a) edge[out=90,in=45,swap] node[prob] {$\tfrac{1}{2}$} (s2)
        (s2a) edge node[prob,swap] {$\tfrac{1}{2}$} (sf)
    ;
\end{tikzpicture}
\end{minipage}%
\begin{minipage}{.3\textwidth}%
\centering%
\begin{tikzpicture}
    \begin{axis}[width=1.1\textwidth,height=3cm,
        xmin=0, xmax=6.5, ymin=0, ymax=0.6, ybar,
        bar width=5pt, bar shift=0pt,
        axis x line=middle,
        axis y line=middle,
        xtick={1,2,3,4,5,6},
        ytick={0.1,0.3,0.5},
        yticklabels={,,.5},
        xticklabels={,,3,4,5,6,},
        x label style={at={(axis description cs:0,0)},anchor=north west,font=\small},
        y label style={at={(axis description cs:0,0)},rotate=90,anchor=south west,font=\small},
        xlabel={\strut value},
        ylabel={\strut prob.},
        clip mode=individual,
        x axis line style=-
    ]
        \addplot [draw=black,fill=white] coordinates {(1,0.5) (2,0.25) (4,0.125) (6,0.0625)};
        \node at (axis description cs:1.08,-.0) {\tiny$\dots$};
    \end{axis}
\end{tikzpicture}
\end{minipage}%
\begin{minipage}{.4\textwidth}%
\centering%
\begin{tikzpicture} %
    \begin{axis}[axis x line=middle, axis y line*=middle,
            height=3cm, width=1.1\textwidth,
            x label style={at={(axis description cs:0,0)},anchor=north west,font=\small},
            y label style={at={(axis description cs:0,0)},anchor=south west,font=\small},
            xmin=0,xmax=7,xtick={1,2,3,4,5,6},xlabel={},
            ymin=-.05,ymax=1.05,ytick={0.5,1},yticklabels={.5,1},ylabel={\hspace{-3pt}cdf}]

\draw[dashed] (axis cs:0,0.00) node[discontinuity-value-small] {} -- (axis cs:1,0.00) node[discontinuity-limit-small] {};
\draw[dashed] (axis cs:1,0.44) node[discontinuity-value-small] {} -- (axis cs:2,0.44) node[discontinuity-limit-small] {};
\draw[dashed] (axis cs:2,0.70) node[discontinuity-value-small] {} -- (axis cs:4,0.70) node[discontinuity-limit-small] {};
\draw[dashed] (axis cs:4,0.84) node[discontinuity-value-small] {} -- (axis cs:8,0.84) node[discontinuity-limit-small] {};
\draw[dashed] (axis cs:8,0.92) node[discontinuity-value-small] {} -- (axis cs:16,0.92) node[discontinuity-limit-small] {};

\draw[gray] (axis cs:0,0.17) -- (axis cs:1,0.17);
\draw[gray] (axis cs:1,0.61) -- (axis cs:2,0.61);
\draw[gray] (axis cs:2,0.87) -- (axis cs:4,0.87);
\draw[gray] (axis cs:4,1.00) -- (axis cs:8,1.00);
\draw[gray] (axis cs:8,1.00) -- (axis cs:16,1.00);

\draw[gray] (axis cs:0,0.00) -- (axis cs:1,0.00);
\draw[gray] (axis cs:1,0.27) -- (axis cs:2,0.27);
\draw[gray] (axis cs:2,0.53) -- (axis cs:4,0.53);
\draw[gray] (axis cs:4,0.67) -- (axis cs:8,0.67);
\draw[gray] (axis cs:8,0.75) -- (axis cs:16,0.75);

\fill[lightgray,fill opacity=.25] (axis cs:0,0.00) -- (axis cs:1,0.00) -- (axis cs:1,0.17) -- (axis cs:0,0.17);
\fill[lightgray,fill opacity=.25] (axis cs:1,0.27) -- (axis cs:2,0.27) -- (axis cs:2,0.61) -- (axis cs:1,0.61);
\fill[lightgray,fill opacity=.25] (axis cs:2,0.53) -- (axis cs:4,0.53) -- (axis cs:4,0.87) -- (axis cs:2,0.87);
\fill[lightgray,fill opacity=.25] (axis cs:4,0.67) -- (axis cs:8,0.67) -- (axis cs:8,1.00) -- (axis cs:4,1.00);
\fill[lightgray,fill opacity=.25] (axis cs:8,0.75) -- (axis cs:16,0.75) -- (axis cs:16,1.00) -- (axis cs:8,1.00);

\draw[thick] (axis cs:0,0) -- (axis cs:1,0) node[discontinuity-limit] {};
\draw[thick] (axis cs:1,.5) node[discontinuity-value] {} -- (axis cs:2,.5) node[discontinuity-limit] {};
\draw[thick] (axis cs:2,.75) node[discontinuity-value] {} -- (axis cs:4,.75) node[discontinuity-limit] {};
\draw[thick] (axis cs:4,.875) node[discontinuity-value] {} -- (axis cs:6,.875) node[discontinuity-limit] {};
\draw[thick] (axis cs:6,.9375) node[discontinuity-value] {} -- (axis cs:6.5,.9375);
\draw[thick,densely dotted] (axis cs:6.5,.9375) -- (axis cs:7,.9375);
    \end{axis}
\end{tikzpicture}
\end{minipage}%
\caption{
    Example of a DTMC (left) together with the probability distribution over possible reward outcomes (middle) and the corresponding cdf (right, solid line).
    The states in the DTMC are (only) labelled by their rewards.
    The right figure also includes an empirical cdf (dashed) and corresponding confidence band (gray) obtained from the DKW inequality (with $\delta = 0.1$ and $k = 50$).
}\label{fig:example}
\end{figure} %

\begin{definition}
\label{def:DTMC}
A \emph{discrete-time Markov chain} (DTMC) is a tuple $\tuple{S, R, T, s_I}$ of a finite set of \emph{states} $S$, a \emph{reward} function $R \colon S \to \RRpos$, an \emph{initial state} $s_I \in S$, and a \emph{transition} function $T\colon S \to \Dist(S)$ mapping each state to a probability distribution over successor states.
A (finite) \emph{path} $\pi$ is (a prefix of) an infinite sequence $\pi = s_0\,s_1\ldots \in S^\omega$ such that $s_0 = s_I$ and $\forall i\colon T(s_i)(s_{i+1}) > 0$.
\end{definition}
See \cref{fig:example} (left) for an example of a DTMC.
A DTMC induces a unique probability measure $\mathbb{P}$ over sets of paths that, intuitively, corresponds to multiplying the probabilities along the path (see \eg\cite[Chapter~10]{BK08}).

\subsubsection{Properties.}
Properties typically consist of two parts:
First, a random variable $X$ assigning a value to each path.
For our results, the choice of $X$ is largely irrelevant; we only require it to yield non-negative finite values. %
Concretely, we consider \emph{total} and \emph{reachability rewards}, \ie $\mathsf{TR}(\pi) = \sum_{i = 0}^\infty R(\pi_i)$ and the same sum cut off at the first goal state, respectively; see \cite[Sec.~2]{BHMWW25} for details, and \cref{fig:example} (middle/right) for the distribution and cdf of $\mathsf{TR}$ on the example DTMC.

Second, a property comes with
an \emph{aggregation function} to \enquote{summarize} $X$ into a single value, traditionally the expected value/mean $\mathbb{E}(X)$ (\wrt $\mathbb{P}$).
Recently, alternative aggregations have gained popularity, for example
\begin{itemize}
    \item higher \emph{moments} (around 0), which are of the form $\mathbb{E}(X^n)$ for $n > 1$;
    \item the $t$-\emph{quantile} (a.k.a.\ the \emph{value-at-risk}) for $t \in (0,1)$, which is the smallest value $v$ such that $X$ is less than or equal to $v$ with probability $t$~\cite{KSBD15,RRS17}:
    $$\Quant_t(X) \defeq \inf\,\set{v \mid \mathbb{P}(X \leq v) \geq t};$$
    \item the \emph{conditional value-at-risk} (a.k.a.\ expected shortfall, expected tail loss, average value-at-risk), which is the expectation over the $t$-quantile, \ie
    $$\CVaR_t(X) \defeq \frac{1}{t} (P \cdot \mathbb{E}[X \mid X < v] + (t - P) \cdot v)$$
    where $t \in (0,1)$, $v = \Quant_t(X)$ and $P = \mathbb{P}(X < v)$~\cite{rockafellar2002conditional,KM18}; and
    \item the \emph{entropic risk}, which with $\gamma > 0$ is~\cite{follmer2002convex,BCMP24}
    $$\ERisk_\gamma(X) \defeq - \frac{1}{\gamma} \log(\mathbb{E}(e^{-\gamma X})).$$
\end{itemize}
We illustrate these for the DTMC of \cref{fig:example} in \ifarxivelse{\Cref{sec:app-example}}{App.~A}.
Ad\-di\-tionally, as in \cite{BHMWW25}, we distinguish whether $X$ has a known upper bound (i.e.\ some $U$ such that $\mathbb{P}(X \leq U) = 1$, the \emph{bounded} case) or not (the \emph{general} case).

\subsubsection{Statistical model checking}\!%
is, at its core, Monte Carlo simulation for formal models and properties:
randomly generate a (predetermined) number $k$ of paths, or \emph{simulations}, from the model that give rise to samples $X_1, \ldots, X_k$ of the random variable~$X$; and from that draw statistical conclusions on the property.
While PMC approaches exist for all of the aforementioned properties, SMC so far exclusively focused on means as follows:
compute the empirical mean
$$\hat{X} \defeq \frac{1}{k}\sum_{i=1}^{k} X_i,$$
and perform a \emph{statistical evaluation} to obtain a confidence interval $I = [l, u] \ni \hat{X}$ at a predetermined \emph{confidence level} $\delta$, so that with (a priori) probability $1 - \delta$ we have $\mathbb{E}(X) \in I$.
That is, if we repeat the SMC procedure $m$ times to obtain confidence intervals $I_1, \ldots, I_m$, we may find some of them (up to $\delta\cdot100\%$ on average) incorrect, i.e.\ $\mathbb{E}(X) \notin I_i$ for some $i$.
Occasionally obtaining an ``incorrect'' result is the nature of a statistical approach based on sampling.
In this work, we develop statistical methods for other aggregations beyond the~mean.

\section{Statistical Guarantees Beyond Means}\label{sec:Estimates}
Before we introduce our approach, we formalise the exact kind of guarantees we aim to give.
Observe that simply returning confidence intervals $[0, \infty]$ is always sound.
However, we also want SMC procedures to yield \enquote{small} intervals.
To formalize this requirement, we say a procedure yields \emph{effective bounds} if (i)~it produces correct intervals with high confidence, and (ii)~for a large enough number $k$ of samples, the intervals produced by the procedure are smaller than any $\varepsilon$ and still correct with high confidence; see \cref{def:bounds} for the formal definition.
We note that this is related to the notion of \emph{consistent} estimators~\cite{Ame85} from statistics, as the mid-point of effective intervals is a consistent estimator.
However, we pose a stronger requirement since we require correct bounds to be produced.

\begin{definition}[effective bounds] \label{def:bounds}
	Let $X$ be a random variable and $\mathcal{F}$ an aggregator, mapping random variables to real numbers.
	An SMC procedure $\mathcal{A}$ yields \emph{effective bounds} on $\mathcal{F}(X)$ if, for any confidence $\delta > 0$, the following two conditions hold:
	(i)~For a collection of independent samples $\Xi$ drawn from $X$, we have $\mathbb{P}(\mathcal{F}(X) \in \mathcal{A}(\Xi, \delta)) \geq 1 - \delta$.
	(ii)~For any precision $\varepsilon > 0$, there exists a threshold $k_0$ such that for a collection of independent samples $\Xi$ drawn from $X$ with $|\Xi| \geq k_0$, we have $\mathbb{P}(\mathcal{F}(X) \in \mathcal{A}(\Xi, \delta) \land |\mathcal{A}(\Xi, \delta)| \leq \varepsilon) \geq 1 - \delta$.
\end{definition}

\begin{remark}\label{rem:dual}
Some works consider the dual problem of gathering enough samples until a given precision is reached.
They seek so-called \emph{probably approximately correct} (PAC) guarantees:
Given confidence level $\delta$ and precision $\varepsilon$, gather enough samples to return $I$ with %
$|I| \leq 2\varepsilon$.
We focus on deriving intervals given a fixed $k$, %
and in \ifarxivelse{\Cref{sec:Sequential}}{App.~B} describe how our methods extend to the dual problem.
\end{remark}
As already observed in \cite{BHMWW25}, obtaining two-sided bounds sometimes is infeasible (depending on the nature of the DTMC, random variable $X$, and aggregator $\mathcal{F}$).
However, we may still be able to derive statistically sound, \enquote{converging} lower bounds.
Thus, we extend the definition of \enquote{limit-PAC} from \cite[Def.~3]{BHMWW25} and say an SMC procedure yields \emph{effective lower bounds} if the value it produces is, with high confidence, (i)~always a lower bound and (ii)~close to the true value if given enough samples.
Formally:

\begin{definition}[effective lower bounds]\label{def:lower_bounds}
	Let $X$ be a random variable and $\mathcal{F}$ an aggregator, mapping random variables to real numbers.
	An SMC procedure $\mathcal{A}$ yields \emph{effective lower bounds} on $\mathcal{F}(X)$ if, for any confidence $\delta > 0$, the following two conditions hold:
	(i)~For a collection of independent samples $\Xi$ drawn from $X$, we have $\mathbb{P}(\mathcal{A}(\Xi, \delta) \leq \mathcal{F}(X)) \geq 1 - \delta$.
	(ii)~For any precision $\varepsilon > 0$, there exists a threshold $k_0$ such that for a collection of independent samples $\Xi$ drawn from $X$ with $|\Xi| \geq k_0$, we have $\mathbb{P}(\mathcal{F}(X) - \varepsilon \leq \mathcal{A}(\Xi, \delta) \leq \mathcal{F}(X)) \geq 1 - \delta$.
\end{definition}

\subsection{DKW: The Dvoretzky-Kiefer-Wolfowitz-Massart Inequality}\label{sec:dkw}
Our key to obtain effective bounds is the \emph{Dvoretzky-Kiefer-Wolfowitz-Massart inequality} (DKW), which relates the cdf of the unknown distribution of $X$ to the \emph{empirical cdf} $\hat{F}(x)=\frac{1}{k} |\set{ X_i \mid X_i \leq x }|$ by
\begin{equation*}
\mathbb{P}\Bigl({\textstyle\sup_{x\in\mathbb R}}\,|\hat{F}(x) - F_X(x)| > \Delta \Bigr) \le \delta \quad \text{where $\Delta = \sqrt{\log(\delta / 2) / (-2k)}$.}
\end{equation*}
Note that $F_X$ is fixed but unknown, while $\hat{F}$ depends on the samples drawn from $X$.
Intuitively, the DKW gives a \emph{confidence band} in which the true cdf lies with high probability; see \Cref{fig:example} (right) for an illustration.
There, $\hat{F}$ is drawn dashed, and the gray area around $\hat{F}$ shows the confidence band (with width $2 \Delta$).
We refer to the bounds of this band as $\underline{F}(x) \defeq \min\,\set{\hat{F}(x) + \Delta,\, 1}$ and $\overline{F}(x) \defeq \max\,\set{0, \,\hat{F}(x) - \Delta}$, respectively.
We denote the random variables that $\underline{F}$, $\hat{F}$, and $\overline{F}$ correspond to as $\underline{X}$, $\hat{X}$, and $\overline{X}$, respectively.
Clearly, $\underline{X} \stochdom \hat{X} \stochdom \overline{X}$ and the DKW implies that $\underline{X} \stochdom X \stochdom \overline{X}$ with high confidence.
In general, $\underline{X}$ and $\overline{X}$ yield $0$ and $\infty$ with probability $\Delta$, respectively.
In the bounded case, we have $F_X(U) = 1$ and hence $\overline{X}$ would instead yield $U$ with probability~$\Delta$.
In \cite{BHMWW25}, this is used to derive (lower and upper, in the bounded case, and lower, in the general case) bounds on expected rewards.

\subsection{Obtaining Effective Bounds}
\label{sec:ObtainingBounds}

As it turns out, computing the aggregations for $\underline{X}$ (and $\overline{X}$) gives effective (lower) bounds for all considered properties. %
We implicitly assume that the DKW condition holds and prove (below) that we then obtain correct (and converging) estimates.
This means in general we get such estimates with high confidence.
Moreover, as all results only depend on the DKW condition holding, we can give guarantees on \emph{all} aggregations simultaneously, without splitting the confidence budget, which is particularly useful for \eg multi-objective queries~\cite{KSBD15,RRS17,KM18}.

\paragraph{Moments.}
For higher-order moments, note that $Y \defeq X^n$ is non-negative and has finite expectation if $X$ satisfies these assumptions.
Thus, the results of \cite[Thm.~1]{BHMWW25} are directly applicable, which state that then the DKW yields effective lower bounds in the general case. 
In the bounded case, we naturally obtain effective bounds by direct application of the DKW (see \Cref{sec:dkw}).

\paragraph{Quantiles.}
By their definition, quantiles are monotone \wrt stochastic dominance, \ie if $Y \stochdom X$, then $\Quant_t(Y) \leq \Quant_t(X)$.
Thus, we also have $\Quant_t(\underline{X}) \leq \Quant_t(X) \leq \Quant_t(\overline{X})$, ensuring correctness of the computed values.
While we can always obtain lower and upper bounds, even in the general case (by choosing $k$ so that $t \in (\Delta, 1 -\Delta]$), only the lower bounds may be effective:
Consider an $X$ with distribution $\{1 \mapsto \frac{1}{2}, 2 \mapsto \frac{1}{2}\}$.
We have $\Quant_{0.5}(X) = 1$, but any sound statistical upper bound on the cdf of $X$ will have $\overline{F}(1) < F_X(1) = 0.5$, and thus always yield a $0.5$-quantile of $2$.
This is a fundamental property of quantiles: they are not continuous \wrt small changes in the distribution.
This already happens for the simple example in \cref{fig:example}, as we illustrate in \ifarxivelse{\cref{sec:app-example}}{App.~A}.
Thus, in general we cannot provide effective bounds.
$X$ is always discrete for DTMC as per \Cref{def:DTMC} and $X = \mathsf{TR}$; we can have non-discrete $X$ if we allow \eg continuously-distributed random rewards, use other models like continuous-time Markov chains (CTMCs), or other properties.
Then, if $X$ is continuous or, at least, if $F_X$ is continuous at $\Quant_t(X)$, we get effective bounds.%

\paragraph{Conditional value-at-risk.}
$\CVaR$ is a \emph{distortion risk measure} (as is $\Quant_t$), which are monotone \wrt stochastic dominance \cite{wirch2001distortion}.
Thus we again immediately get $\CVaR_t(\underline{X}) \leq \CVaR_t(X) \leq \CVaR_t(\overline{X})$.
In contrast to general expectations, the bounded and general case do not differ:
By assumption, we have $X < \infty$, hence there exists $T$ such that $F_X(T) > 1 - \frac{t}{2}$.
For a large enough $k$, we have $\Delta < \frac{t}{2}$, and $\overline{F}(T) \geq 1 - t$.
Then, we know (with high confidence) that $X \leq T$ with probability $t$, i.e.\ $\Quant_t(\overline{X}) \leq T$ and therefore $\CVaR_t(\overline{X}) \leq T < \infty$.
Thus, we can directly bound $|\!\CVaR_t(\underline{X}) - \CVaR_t(\overline{X})|$ by $T \cdot 2 \Delta$, which goes to $0$ for large enough $k$.
Together, we obtain effective bounds in the general case.

\paragraph{Entropic risk.}
First, observe that if $Y \stochdom X$, then $e^{-\gamma X} \stochdom e^{-\gamma Y}$ (the order reverses as $e^{-\gamma x}$ is decreasing).
Consequently, $\mathbb{E}(e^{-\gamma X}) \leq \mathbb{E}(e^{- \gamma Y})$, and thus $\ERisk_\gamma(Y) \leq \ERisk_\gamma(X)$ (recall that $\ERisk = - {1}/{\gamma} \cdot \dots$).
Hence, we get $\ERisk_\gamma(\underline{X}) \leq \ERisk_\gamma(X) \leq \ERisk_\gamma(\overline{X})$.
While there is no strict \enquote{cut-off} as for $\CVaR$, we argue that we can still bound the overall difference between $X$ and the bounds $\underline{X}$ and $\overline{X}$ in general.
We have
$$\ERisk_\gamma(\overline{X}) - \ERisk_\gamma(X) = -{1}/{\gamma}\cdot\log(\mathbb{E}(e^{-\gamma \overline{X}}) / \mathbb{E}(e^{-\gamma X})).$$
We now apply two useful general facts about cdfs, namely that
(i)~$\mathbb{E}(X) = \int_x (1 - F_X(x)) \cdot x \,dx$ and
(ii)~for a positive, continuous, strictly decreasing function $f$ we have $F_{f(X)}(x) = 1 - F_X(f^{-1}(x))$.
We get\\
\centerline{
$\mathbb{E}(e^{-\gamma \overline{X}}) / \mathbb{E}(e^{-\gamma X}) = \int \overline{F}(- \frac{1}{\gamma} \log(x)) e^{-\gamma x} \,dx / \mathbb{E}(e^{-\gamma X}).$}\\
Recall that $F_X(x) - \Delta \leq \overline{F}(x)$.
Hence,\\[3pt]
\centerline{
$
\begin{aligned}
\textstyle\int \overline{F}(- \frac{1}{\gamma} \log(x)) e^{-\gamma x} \,dx / \mathbb{E}(e^{-\gamma X})
&\textstyle\geq \int (F_X(- \frac{1}{\gamma} \log(x)) - \Delta) e^{-\gamma x} \,dx / \mathbb{E}(e^{-\gamma X})\\[-3pt]
&\textstyle= 1 - \Delta \int e^{-\gamma x} \, dx / \mathbb{E}(e^{-\gamma X}).
\end{aligned}
$}\\
Consequently, for $\Delta \to 0$ this expression converges to $1$, and thus $\ERisk_\gamma(\overline{X}) - \ERisk_\gamma(X) = -\frac{1}{\gamma} \log({\dots}) \to 0$.
The proof for $\underline{X}$ is analogous.

\section{Tool Implementation}
\label{sec:Implementation}

The \tool{modes} SMC tool~\cite{BDHS20} was recently extended with sound statistical methods for estimating means, including the DKW~\cite{BHMWW25}.
Now, in version 3.1.287, we added syntax for quantile and CVaR properties to the parsers for its input languages, \lang{Modest}~\cite{HHHK13,BDHK06} and \lang{Jani}~\cite{BDHHJT17}, and extended its implementation of the DKW to estimate and provide bounds for such properties. %
Additionally, empirical cdfs can be exported to CSV and Excel files for plotting and further analysis.

\begin{table}[t]
\centering
\caption{Estimates and DKW confidence intervals for the examples' properties.}
\label{tab:PropertyResults}
\setlength{\tabcolsep}{6pt}
\begin{tabular}{lrrcrcrc}
\toprule
&& \multicolumn{2}{c}{expected value} & \multicolumn{2}{c}{$0.3$-quantile} & \multicolumn{2}{c}{$\CVaR_{0.3}$} \\
\cmidrule(lr){3-4}\cmidrule(lr){5-6}\cmidrule(lr){7-8}
example & $k$ & $\hat{X}$ & conf.\ int. & est. & conf.\ int. & est. & conf.\ int.\\\midrule
\textit{coupon}           &  $100$ & $13.11$ &   $[10.08, \infty)$ & $10$    & $\pz[9,   11]$   & $8.90$ & $[4.56,   10.14]$\\
                          & $1000$ & $12.90$ &   $[11.82, \infty)$ & $10$    &   $[10,   11]$   & $8.71$ & $[7.28, \pz9.21]$\\[1.5pt]
\textit{leader\!\_\!sync} &  $100$ &  $1.17$ & \pz$[0.89, \infty)$ &  $1$    & $\pz[1, 1]\pz$   & $1.00$ & $[0.55, \pz1.00]$\\[1.5pt]
\textit{embedded}         &  $100$ &  $0.35$ &       n/a           &  $0.13$ & $[0.10, 0.20]$   & $0.35$ &      n/a       \\
                          & $1000$ &  $0.33$ &       n/a           &  $0.17$ & $[0.15, 0.19]$   & $0.33$ &      n/a       \\
\bottomrule
\end{tabular}
\end{table}

To demonstrate the new tool features, we use three examples from the QVBS selected for diversity in cdfs:
(1)~the \textit{coupon} model with parameters $\texttt{N} = 15$, $\texttt{DRAWS} = 4$, $\texttt{B} = 5$ (a DTMC of 17 billion states, to which SMC is agnostic) and the random variable underlying property \textit{exp\_\!draws};
(2)~\textit{leader\!\_\!sync} with $\texttt{N} = 5$, $\texttt{K} = 4$ (DTMC, 4244 states) and \textit{time}; and
(3)~\textit{embedded} with $\texttt{MAX\_COUNT} = 8$, $\texttt{T} = 12$ (a continuous-time Markov chain of 8548 states) and \textit{danger\!\_\!time}.
The original properties query for expected reachability rewards; we add properties querying for the $0.3$-quantile and $\CVaR_{0.3}$ of the same reward specification, \ie the same random variable on paths.
We run \tool{modes} on each example with $k = 100$ simulations, and on \textit{coupon} and \textit{embedded} additionally with $k = 1000$.
In addition to obtaining DKW-based confidence intervals, we use \tool{modes}' new \texttt{-{}-cdf} parameter to export empirical CDFs with DKW confidence bands.

In \Cref{tab:PropertyResults}, we show the results that \tool{modes} obtains for the properties.
As reachability rewards fall into the general case, we can only obtain lower bounds for the expected values~\cite{BHMWW25}.
For quantiles and CVaR, as per \Cref{sec:ObtainingBounds}, the DKW allows us to obtain (for CVaR effective) lower and upper bounds.
On \textit{embedded}, \tool{modes} cannot apply the DKW to expectation and CVaR because its syntactic procedure to find a lower bound for the rewards fails as they are encoded via an unbounded real-valued variable.
For quantiles, the absence of bounds on the distribution is no hindrance.
The DKW can produce rather asymmetric confidence intervals, which we see for \textit{leader\!\_\!sync}'s CVaR property.

\newcommand*{\GetMinMax}[4]{%
    \pgfmathsetmacro#3{-16383}%
    \pgfmathsetmacro#4{16383}%
    \pgfplotstableforeachcolumnelement{[index]#2}\of{#1}\as\cellValue{%
        \ifx\cellValue\@empty\else
            \pgfmathsetmacro{#3}{max(#3,\cellValue)}%
            \pgfmathsetmacro{#4}{min(#4,\cellValue)}%
        \fi
    }
}
\newcommand{\cdfplot}[3]{
\pgfplotstableread[col sep=comma]{#1}\cdfplottable
\GetMinMax{\cdfplottable}{0}{\cdfplottablemax}{\cdfplottablemin}
\pgfplotsset{
  width=4.67cm,
  every axis plot/.append style={mark=none},
  every tick label/.append style={font=\tiny}
}
\begin{tikzpicture}
\begin{axis}[
  xmin=0, xmax=\cdfplottablemax, ymin=0, ymax=1,
  enlargelimits=0,
  legend style={font=\scriptsize,at={(1.36,1)},anchor=north west}, legend cell align={left},
  xlabel={#2}
]
\addplot[DarkOrchid, const plot, line width=1pt] table [x index=0, y index=3] {\cdfplottable};
\addplot[BurntOrange, const plot, line width=1pt] table [x index=0, y index=2] {\cdfplottable};
\addplot[BlueGreen, const plot, line width=1pt] table [x index=0, y index=1] {\cdfplottable};
\addplot[BrickRed] coordinates { (0,0.3) (\cdfplottablemax,0.3) };
#3
\end{axis}
\end{tikzpicture}
}

\begin{figure}[t]
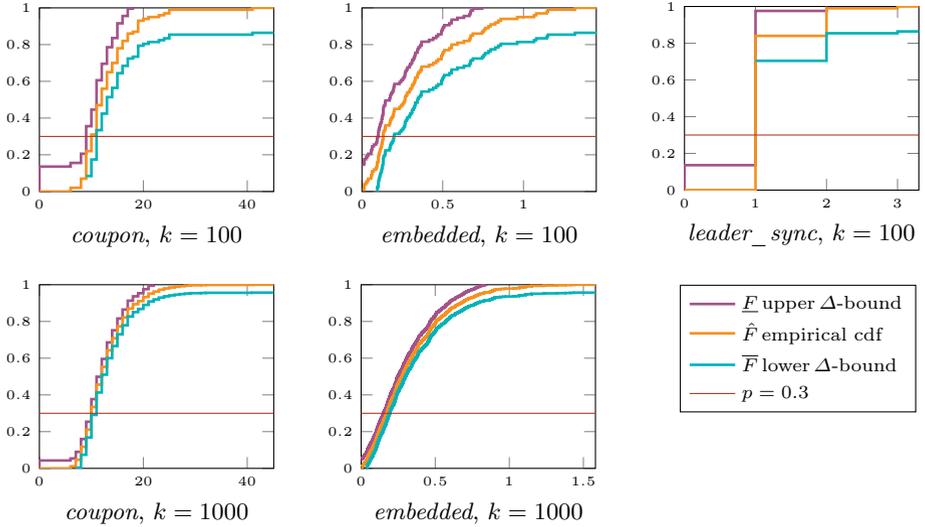

\centering
\cdfplot{experiments/coupon.15-4-5.100.csv}{\textit{coupon}, $k=100$}{}
\cdfplot{experiments/embedded.8-12.100.csv}{\textit{embedded}, $k=100$}{}
\cdfplot{experiments/leader_sync.5-4.100.csv}{\textit{leader\_sync}, $k=100$}{}\\[2mm]
\cdfplot{experiments/coupon.15-4-5.1000.csv}{\textit{coupon}, $k=1000$}{}
\cdfplot{experiments/embedded.8-12.1000.csv}{\textit{embedded}, $k=1000$}{
  \legend{$\underline{F}\!$ upper $\!\Delta$-bound,$\hat{F}\!$ empirical cdf,$\overline{F}\!$ lower $\!\Delta$-bound,$p = 0.3$}
}\vspace{-3pt}
\caption{DKW CDF confidence bands obtained on the example benchmarks.}
\label{fig:CdfBenchmarks}
\end{figure}

We plot the empirical cdfs and associated DKW confidence bounds \tool{modes} delivered in \Cref{fig:CdfBenchmarks}.
Graphically, the confidence interval for the quantile is the $p=0.3$ line's segment between the bound curves, while the CVaR estimate and confidence interval stem from the curves cut off at that line.
The DTMCs' reward distributions are necessarily discrete:
\textit{leader\!\_\!sync} has only $3$ possible outcomes of non-negligible probability (we thus omit $k=1000$);
increasing $k$ for \textit{coupon} does not smoothen the curve much, because most elements of the distribution's support were already sampled at $k=100$---only the confidence band gets much thinner.
For the CTMC \textit{embedded}, the distribution is visibly continuous.

We do not report runtimes because the overhead of using the DKW---for collecting all samples to finally compute the intervals, instead of incremental averaging plus evaluation based on $k$ and $\delta$ only as for traditional methods for the mean---was negligible in these experiments.
Aside from DTMCs and CTMCs, \tool{modes} also supports more complex and expressive formalisms up to stochastic hybrid automata~\cite{FHHWZ11}; the new methods to check and bound quantiles and CVaRs work independent of the model type.
They equally combine orthogonally with \tool{modes}' features for rare event simulation~\cite{BDH19} as well as learning and scheduler sampling for nondeterministic models~\cite{DHS18,DFHR25,NHR21}.

\section{Conclusion}

In this work, we have shown how the DKW inequality can be used to derive bounds on various aggregation functions beyond the classical expectation/mean, closing a significant gap of SMC compared to traditional verification.
Moreover, as all our estimations are based on the DKW inequality, our methods can estimate all values simultaneously.
Our experimental evaluation confirms the effectiveness of our methods, allowing for scalable estimation of such aggregation values for large systems.
For future work, we believe that our approach should also be applicable to other risk measures such as variance~\cite{DBLP:conf/icml/MannorT11,DBLP:journals/jcss/BrazdilCFK17}, variance-penalized expected payoff~\cite{markowitz1991foundations,DBLP:conf/icalp/PiribauerSB22,DBLP:journals/eor/MaMX23,DBLP:conf/concur/BaierPS24}, or cumulative prospect theory~\cite{CPT92,BCMW25}.

\medskip
\noindent
\textbf{Data availability.}
The \tool{modes} tool is available at \href{https://www.modestchecker.net/Downloads/}{modestchecker.net}.
An artifact for this paper---a reproduction package with the models and commands for the experiments in \Cref{sec:Implementation}---is available at DOI\:\,\href{https://zenodo.org/records/15286509}{10.5281/zenodo.15286509}.

\bibliographystyle{splncs04}
\bibliography{paper}

\begin{thebibliography}{10}
\providecommand{\url}[1]{\texttt{#1}}
\providecommand{\urlprefix}{URL }
\providecommand{\doi}[1]{https://doi.org/#1}

\bibitem{AP18}
Agha, G., Palmskog, K.: A survey of statistical model checking. {ACM} Trans.
  Model. Comput. Simul.  \textbf{28}(1),  6:1--6:39 (2018).
  \doi{10.1145/3158668}

\bibitem{Ame85}
Amemiya, T.: Advanced econometrics. Harvard university press (1985)

\bibitem{And69}
Anderson, T.W.: Confidence limits for the value of an arbitrary bounded random
  variable with a continuous distribution function. Bulletin of The
  International and Statistical Institute  \textbf{43},  249--251 (1969)

\bibitem{Bai16}
Baier, C.: Probabilistic model checking. In: Esparza, J., Grumberg, O.,
  Sickert, S. (eds.) Dependable Software Systems Engineering, {NATO} Science
  for Peace and Security Series -- {D}: Information and Communication Security,
  vol.~45, pp. 1--23. {IOS} Press (2016). \doi{10.3233/978-1-61499-627-9-1}

\bibitem{BAFK18}
Baier, C., de~Alfaro, L., Forejt, V., Kwiatkowska, M.: Model checking
  probabilistic systems. In: Clarke, E.M., Henzinger, T.A., Veith, H., Bloem,
  R. (eds.) Handbook of Model Checking, pp. 963--999. Springer (2018).
  \doi{10.1007/978-3-319-10575-8_28}

\bibitem{BCMP24}
Baier, C., Chatterjee, K., Meggendorfer, T., Piribauer, J.: Entropic risk for
  turn-based stochastic games. Inf. Comput.  \textbf{301},  105214 (2024).
  \doi{10.1016/J.IC.2024.105214},
  \url{https://doi.org/10.1016/j.ic.2024.105214}

\bibitem{BK08}
Baier, C., Katoen, J.P.: Principles of model checking. {MIT} Press (2008),
  \url{https://mitpress.mit.edu/books/principles-model-checking}

\bibitem{DBLP:conf/concur/BaierPS24}
Baier, C., Piribauer, J., Starke, M.: Risk-averse optimization of total rewards
  in {M}arkovian models using deviation measures. In: {CONCUR}. LIPIcs,
  vol.~311, pp. 9:1--9:20. Schloss Dagstuhl - Leibniz-Zentrum f{\"{u}}r
  Informatik (2024)

\bibitem{BDHK06}
Bohnenkamp, H.C., D'Argenio, P.R., Hermanns, H., Katoen, J.P.: {MoDeST}: A
  compositional modeling formalism for hard and softly timed systems. {IEEE}
  Trans. Software Eng.  \textbf{32}(10),  812--830 (2006).
  \doi{10.1109/TSE.2006.104}

\bibitem{DBLP:journals/jcss/BrazdilCFK17}
Br{\'{a}}zdil, T., Chatterjee, K., Forejt, V., Kucera, A.: Trading performance
  for stability in {M}arkov decision processes. J. Comput. Syst. Sci.
  \textbf{84},  144--170 (2017)

\bibitem{BCMW25}
Brihaye, T., Chatterjee, K., Mohr, S., Weininger, M.: Risk-aware {M}arkov
  decision processes using cumulative prospect theory. In: {LICS} (to appear,
  preprint available at \url{https://arxiv.org/abs/2505.09514}) (2025)

\bibitem{BDH19}
Budde, C.E., D'Argenio, P.R., Hartmanns, A.: Automated compositional importance
  splitting. Sci. Comput. Program.  \textbf{174},  90--108 (2019).
  \doi{10.1016/J.SCICO.2019.01.006}

\bibitem{BDHS20}
Budde, C.E., D'Argenio, P.R., Hartmanns, A., Sedwards, S.: An efficient
  statistical model checker for nondeterminism and rare events. Int. J. Softw.
  Tools Technol. Transf.  \textbf{22}(6),  759--780 (2020).
  \doi{10.1007/S10009-020-00563-2}

\bibitem{BDHHJT17}
Budde, C.E., Dehnert, C., Hahn, E.M., Hartmanns, A., Junges, S., Turrini, A.:
  {JANI}: Quantitative model and tool interaction. In: Legay, A., Margaria, T.
  (eds.) 23rd International Conference on Tools and Algorithms for the
  Construction and Analysis of Systems ({TACAS}). Lecture Notes in Computer
  Science, vol. 10206, pp. 151--168 (2017). \doi{10.1007/978-3-662-54580-5\_9}

\bibitem{BHMWW25}
Budde, C.E., Hartmanns, A., Meggendorfer, T., Weininger, M., Wienh{\"{o}}ft,
  P.: Sound statistical model checking for probabilities and expected rewards.
  In: {TACAS} {(1)}. Lecture Notes in Computer Science, vol. 15696, pp.
  167--190. Springer (2025). \doi{10.1007/978-3-031-90643-5\_9},
  \url{https://doi.org/10.1007/978-3-031-90643-5\_9}

\bibitem{CP34}
Clopper, C., Pearson, E.: The use of confidence or fiducial limits illustrated
  in the case of the binomial. Biometrika  \textbf{26}(4),  404--413 (1934).
  \doi{10.1093/biomet/26.4.404}

\bibitem{DFHR25}
D'Argenio, P.R., Fraire, J.A., Hartmanns, A., Raverta, F.D.: Comparing
  statistical, analytical, and learning-based routing approaches for
  delay-tolerant networks. {ACM} Trans. Model. Comput. Simul.  \textbf{35}(2)
  (2025). \doi{10.1145/3665927}

\bibitem{DHS18}
D'Argenio, P.R., Hartmanns, A., Sedwards, S.: Lightweight statistical model
  checking in nondeterministic continuous time. In: Margaria, T., Steffen, B.
  (eds.) 8th International Symposium on Leveraging Applications of Formal
  Methods, Verification and Validation ({ISoLA}). Lecture Notes in Computer
  Science, vol. 11245, pp. 336--353. Springer (2018).
  \doi{10.1007/978-3-030-03421-4\_22}

\bibitem{DM18}
D'Argenio, P.R., Monti, R.E.: Input/output stochastic automata with urgency:
  Confluence and weak determinism. In: Fischer, B., Uustalu, T. (eds.) 15th
  International Colloquium on Theoretical Aspects of Computing ({ICTAC}).
  Lecture Notes in Computer Science, vol. 11187, pp. 132--152. Springer (2018).
  \doi{10.1007/978-3-030-02508-3\_8}

\bibitem{DLLMW11}
David, A., Larsen, K.G., Legay, A., Mikucionis, M., Wang, Z.: Time for
  statistical model checking of real-time systems. In: Gopalakrishnan, G.,
  Qadeer, S. (eds.) 23rd International Conference on Computer Aided
  Verification ({CAV}). Lecture Notes in Computer Science, vol.~6806, pp.
  349--355. Springer (2011). \doi{10.1007/978-3-642-22110-1\_27}

\bibitem{DKW56}
Dvoretzky, A., Kiefer, J., Wolfowitz, J.: Asymptotic minimax character of the
  sample distribution function and of the classical multinomial estimator. The
  Annals of Mathematical Statistics  \textbf{27}(3),  642--669 (1956).
  \doi{10.1214/aoms/1177728174}

\bibitem{EGF15}
Ellen, C., Gerwinn, S., Fr{\"{a}}nzle, M.: Statistical model checking for
  stochastic hybrid systems involving nondeterminism over continuous domains.
  Int. J. Softw. Tools Technol. Transf.  \textbf{17}(4),  485--504 (2015).
  \doi{10.1007/S10009-014-0329-Y}

\bibitem{follmer2002convex}
F{\"o}llmer, H., Schied, A.: Convex measures of risk and trading constraints.
  Finance and stochastics  \textbf{6},  429--447 (2002)

\bibitem{FHHWZ11}
Fr{\"{a}}nzle, M., Hahn, E.M., Hermanns, H., Wolovick, N., Zhang, L.:
  Measurability and safety verification for stochastic hybrid systems. In:
  Caccamo, M., Frazzoli, E., Grosu, R. (eds.) 14th {ACM} International
  Conference on Hybrid Systems: Computation and Control ({HSCC}). pp. 43--52.
  {ACM} (2011). \doi{10.1145/1967701.1967710}

\bibitem{HHHK13}
Hahn, E.M., Hartmanns, A., Hermanns, H., Katoen, J.P.: A compositional
  modelling and analysis framework for stochastic hybrid systems. Formal
  Methods Syst. Des.  \textbf{43}(2),  191--232 (2013).
  \doi{10.1007/S10703-012-0167-Z}

\bibitem{HH14}
Hartmanns, A., Hermanns, H.: The {M}odest {T}oolset: An integrated environment
  for quantitative modelling and verification. In: {\'{A}}brah{\'{a}}m, E.,
  Havelund, K. (eds.) 20th International Conference on Tools and Algorithms for
  the Construction and Analysis of Systems ({TACAS}). Lecture Notes in Computer
  Science, vol.~8413, pp. 593--598. Springer (2014).
  \doi{10.1007/978-3-642-54862-8_51}

\bibitem{HJQW23}
Hartmanns, A., Junges, S., Quatmann, T., Weininger, M.: A practitioner's guide
  to {MDP} model checking algorithms. In: Sankaranarayanan, S., Sharygina, N.
  (eds.) 29th International Conference on Tools and Algorithms for the
  Construction and Analysis of Systems ({TACAS}). Lecture Notes in Computer
  Science, vol. 13993, pp. 469--488. Springer (2023).
  \doi{10.1007/978-3-031-30823-9\_24}

\bibitem{HKPQR19}
Hartmanns, A., Klauck, M., Parker, D., Quatmann, T., Ruijters, E.: The
  quantitative verification benchmark set. In: Vojnar, T., Zhang, L. (eds.)
  25th International Conference on Tools and Algorithms for the Construction
  and Analysis of Systems ({TACAS}). Lecture Notes in Computer Science, vol.
  11427, pp. 344--350. Springer (2019). \doi{10.1007/978-3-030-17462-0_20}

\bibitem{Hoe63}
Hoeffding, W.: Probability inequalities for sums of bounded random variables.
  Journal of the American Statistical Association  \textbf{58}(301),  13--30
  (1963). \doi{10.1080/01621459.1963.10500830}

\bibitem{KSBD15}
Kr{\"{a}}hmann, D., Schubert, J., Baier, C., Dubslaff, C.: Ratio and weight
  quantiles. In: Italiano, G.F., Pighizzini, G., Sannella, D. (eds.) 40th
  International Symposium on Mathematical Foundations of Computer Science 2015
  ({MFCS}). Lecture Notes in Computer Science, vol.~9234, pp. 344--356.
  Springer (2015). \doi{10.1007/978-3-662-48057-1\_27}

\bibitem{Kre16}
Kret{\'{\i}}nsk{\'{y}}, J.: Survey of statistical verification of linear
  unbounded properties: Model checking and distances. In: ISoLA {(1)}. Lecture
  Notes in Computer Science, vol.~9952, pp. 27--45 (2016).
  \doi{10.1007/978-3-319-47166-2_3}

\bibitem{KM18}
Kret{\'{\i}}nsk{\'{y}}, J., Meggendorfer, T.: Conditional value-at-risk for
  reachability and mean payoff in markov decision processes. In: Dawar, A.,
  Gr{\"{a}}del, E. (eds.) 33rd Annual {ACM/IEEE} Symposium on Logic in Computer
  Science ({LICS}). pp. 609--618. {ACM} (2018). \doi{10.1145/3209108.3209176}

\bibitem{LLTYSG19}
Legay, A., Lukina, A., Traonouez, L.M., Yang, J., Smolka, S.A., Grosu, R.:
  Statistical model checking. In: Steffen, B., Woeginger, G.J. (eds.) Computing
  and Software Science -- State of the Art and Perspectives, Lecture Notes in
  Computer Science, vol. 10000, pp. 478--504. Springer (2019).
  \doi{10.1007/978-3-319-91908-9_23}

\bibitem{DBLP:journals/eor/MaMX23}
Ma, S., Ma, X., Xia, L.: A unified algorithm framework for mean-variance
  optimization in discounted {M}arkov decision processes. Eur. J. Oper. Res.
  \textbf{311}(3),  1057--1067 (2023)

\bibitem{DBLP:conf/icml/MannorT11}
Mannor, S., Tsitsiklis, J.N.: Mean-variance optimization in {M}arkov decision
  processes. In: {ICML}. pp. 177--184. Omnipress (2011)

\bibitem{markowitz1991foundations}
Markowitz, H.M.: Foundations of portfolio theory. The journal of finance
  \textbf{46}(2),  469--477 (1991)

\bibitem{Mas90}
Massart, P.: The tight constant in the {D}voretzky-{K}iefer-{W}olfowitz
  inequality. The Annals of Probability  \textbf{18}(3),  1269--1283 (1990).
  \doi{10.1214/aop/1176990746}

\bibitem{NHR21}
Niehage, M., Hartmanns, A., Remke, A.: Learning optimal decisions for
  stochastic hybrid systems. In: Arun-Kumar, S., M{\'{e}}ry, D., Saha, I.,
  Zhang, L. (eds.) 19th {ACM-IEEE} International Conference on Formal Methods
  and Models for System Design ({MEMOCODE}). pp. 44--55. {ACM} (2021).
  \doi{10.1145/3487212.3487339}

\bibitem{PR17}
Pilch, C., Remke, A.: Statistical model checking for hybrid {P}etri nets with
  multiple general transitions. In: 47th Annual {IEEE/IFIP} International
  Conference on Dependable Systems and Networks ({DSN}). pp. 475--486. {IEEE}
  Computer Society (2017). \doi{10.1109/DSN.2017.41}

\bibitem{DBLP:conf/icalp/PiribauerSB22}
Piribauer, J., Sankur, O., Baier, C.: The variance-penalized stochastic
  shortest path problem. In: {ICALP}. LIPIcs, vol.~229, pp. 129:1--129:19.
  Schloss Dagstuhl - Leibniz-Zentrum f{\"{u}}r Informatik (2022)

\bibitem{RRS17}
Randour, M., Raskin, J.F., Sankur, O.: Percentile queries in multi-dimensional
  {M}arkov decision processes. Formal Methods Syst. Des.  \textbf{50}(2-3),
  207--248 (2017). \doi{10.1007/S10703-016-0262-7}

\bibitem{rockafellar2002conditional}
Rockafellar, R.T., Uryasev, S.: Conditional value-at-risk for general loss
  distributions. Journal of banking \& finance  \textbf{26}(7),  1443--1471
  (2002)

\bibitem{CPT92}
Tversky, A., Kahneman, D.: Advances in prospect theory: Cumulative
  representation of uncertainty. Journal of Risk and uncertainty
  \textbf{5}(4),  297--323 (1992)

\bibitem{UB13}
Ummels, M., Baier, C.: Computing quantiles in markov reward models. In:
  Pfenning, F. (ed.) Foundations of Software Science and Computation Structures
  - 16th International Conference, {FOSSACS} 2013, Held as Part of the European
  Joint Conferences on Theory and Practice of Software, {ETAPS} 2013, Rome,
  Italy, March 16-24, 2013. Proceedings. Lecture Notes in Computer Science,
  vol.~7794, pp. 353--368. Springer (2013). \doi{10.1007/978-3-642-37075-5\_23}

\bibitem{Wal45}
Wald, A.: Sequential tests of statistical hypotheses. The Annals of
  Mathematical Statistics  \textbf{16}(2),  117--186 (1945).
  \doi{10.1214/aoms/1177731118}

\bibitem{wirch2001distortion}
Wirch, J.L., Hardy, M.R.: Distortion risk measures: Coherence and stochastic
  dominance. In: International congress on insurance: Mathematics and
  economics. pp. 15--17. Citeseer (2001)

\bibitem{YS02}
Younes, H.L.S., Simmons, R.G.: Probabilistic verification of discrete event
  systems using acceptance sampling. In: Brinksma, E., Larsen, K.G. (eds.) 14th
  International Conference on Computer Aided Verification ({CAV}). Lecture
  Notes in Computer Science, vol.~2404, pp. 223--235. Springer (2002).
  \doi{10.1007/3-540-45657-0_17}

\end{thebibliography}

\iftoggle{arxiv}{
    \newpage
    \appendix    
    \section{Example: Computing Quantities of Interest} \label{sec:app-example}

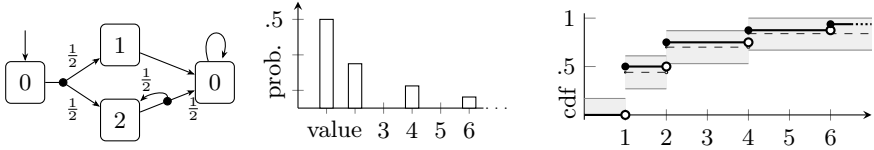
\begin{figure}[t]
    \centering%
\begin{minipage}{.275\textwidth}
\centering
\begin{tikzpicture}[auto,xscale=1.25]
    \node[state,initial above,initial text=] at (0,0) (init) {$0$};
    \node[state] at (1,.5) (s1) {$1$};
    \node[state] at (1,-.5) (s2) {$2$};
    \node[state] at (2,0) (sf) {$0$};
    \node[actionnode] at (.4,0) (inita) {};
    \node[actionnode] at (1.5,-.25) (s2a) {};

    \path[actionedge]
        (init) edge (inita)
        (s2) edge (s2a)
    ;
    \path[directedge]
        (s1) edge (sf)
        (sf) edge[loop above] (sf)
    ;
    \path[probedge]
        (inita) edge node[prob] {$\tfrac{1}{2}$} (s1)
        (inita) edge[swap] node[prob] {$\tfrac{1}{2}$} (s2)
        (s2a) edge[out=90,in=45,swap] node[prob] {$\tfrac{1}{2}$} (s2)
        (s2a) edge node[prob,swap] {$\tfrac{1}{2}$} (sf)
    ;
\end{tikzpicture}
\end{minipage}%
\begin{minipage}{.3\textwidth}%
\centering%
\begin{tikzpicture}
    \begin{axis}[width=1.1\textwidth,height=3cm,
        xmin=0, xmax=6.5, ymin=0, ymax=0.6, ybar,
        bar width=5pt, bar shift=0pt,
        axis x line=middle,
        axis y line=middle,
        xtick={1,2,3,4,5,6},
        ytick={0.1,0.3,0.5},
        yticklabels={,,.5},
        xticklabels={,,3,4,5,6,},
        x label style={at={(axis description cs:0,0)},anchor=north west,font=\small},
        y label style={at={(axis description cs:0,0)},rotate=90,anchor=south west,font=\small},
        xlabel={\strut value},
        ylabel={\strut prob.},
        clip mode=individual,
        x axis line style=-
    ]
        \addplot [draw=black,fill=white] coordinates {(1,0.5) (2,0.25) (4,0.125) (6,0.0625)};
        \node at (axis description cs:1.08,-.0) {\tiny$\dots$};
    \end{axis}
\end{tikzpicture}
\end{minipage}%
\begin{minipage}{.4\textwidth}%
\centering%
\begin{tikzpicture} %
    \begin{axis}[axis x line=middle, axis y line*=middle,
            height=3cm, width=1.1\textwidth,
            x label style={at={(axis description cs:0,0)},anchor=north west,font=\small},
            y label style={at={(axis description cs:0,0)},anchor=south west,font=\small},
            xmin=0,xmax=7,xtick={1,2,3,4,5,6},xlabel={},
            ymin=-.05,ymax=1.05,ytick={0.5,1},yticklabels={.5,1},ylabel={\hspace{-3pt}cdf}]

\draw[dashed] (axis cs:0,0.00) node[discontinuity-value-small] {} -- (axis cs:1,0.00) node[discontinuity-limit-small] {};
\draw[dashed] (axis cs:1,0.44) node[discontinuity-value-small] {} -- (axis cs:2,0.44) node[discontinuity-limit-small] {};
\draw[dashed] (axis cs:2,0.70) node[discontinuity-value-small] {} -- (axis cs:4,0.70) node[discontinuity-limit-small] {};
\draw[dashed] (axis cs:4,0.84) node[discontinuity-value-small] {} -- (axis cs:8,0.84) node[discontinuity-limit-small] {};
\draw[dashed] (axis cs:8,0.92) node[discontinuity-value-small] {} -- (axis cs:16,0.92) node[discontinuity-limit-small] {};

\draw[gray] (axis cs:0,0.17) -- (axis cs:1,0.17);
\draw[gray] (axis cs:1,0.61) -- (axis cs:2,0.61);
\draw[gray] (axis cs:2,0.87) -- (axis cs:4,0.87);
\draw[gray] (axis cs:4,1.00) -- (axis cs:8,1.00);
\draw[gray] (axis cs:8,1.00) -- (axis cs:16,1.00);

\draw[gray] (axis cs:0,0.00) -- (axis cs:1,0.00);
\draw[gray] (axis cs:1,0.27) -- (axis cs:2,0.27);
\draw[gray] (axis cs:2,0.53) -- (axis cs:4,0.53);
\draw[gray] (axis cs:4,0.67) -- (axis cs:8,0.67);
\draw[gray] (axis cs:8,0.75) -- (axis cs:16,0.75);

\fill[lightgray,fill opacity=.25] (axis cs:0,0.00) -- (axis cs:1,0.00) -- (axis cs:1,0.17) -- (axis cs:0,0.17);
\fill[lightgray,fill opacity=.25] (axis cs:1,0.27) -- (axis cs:2,0.27) -- (axis cs:2,0.61) -- (axis cs:1,0.61);
\fill[lightgray,fill opacity=.25] (axis cs:2,0.53) -- (axis cs:4,0.53) -- (axis cs:4,0.87) -- (axis cs:2,0.87);
\fill[lightgray,fill opacity=.25] (axis cs:4,0.67) -- (axis cs:8,0.67) -- (axis cs:8,1.00) -- (axis cs:4,1.00);
\fill[lightgray,fill opacity=.25] (axis cs:8,0.75) -- (axis cs:16,0.75) -- (axis cs:16,1.00) -- (axis cs:8,1.00);

\draw[thick] (axis cs:0,0) -- (axis cs:1,0) node[discontinuity-limit] {};
\draw[thick] (axis cs:1,.5) node[discontinuity-value] {} -- (axis cs:2,.5) node[discontinuity-limit] {};
\draw[thick] (axis cs:2,.75) node[discontinuity-value] {} -- (axis cs:4,.75) node[discontinuity-limit] {};
\draw[thick] (axis cs:4,.875) node[discontinuity-value] {} -- (axis cs:6,.875) node[discontinuity-limit] {};
\draw[thick] (axis cs:6,.9375) node[discontinuity-value] {} -- (axis cs:6.5,.9375);
\draw[thick,densely dotted] (axis cs:6.5,.9375) -- (axis cs:7,.9375);
    \end{axis}
\end{tikzpicture}
\end{minipage}%
\caption{
    Example of a DTMC (left) together with the probability distribution over possible reward outcomes (middle) and the corresponding cdf (right, solid line).
    The states in the DTMC are (only) labelled by their rewards.
    The right figure also includes an empirical cdf (dashed) and corresponding confidence band (gray) obtained from the DKW inequality (with $\delta = 0.1$ and $k = 50$).
}\label{fig:example_app}
\end{figure}

For convenience, we repeat \Cref{fig:example} here in \Cref{fig:example_app}.
In this appendix, we show how compute the quantities of interest listed in \Cref{sec:prelims} for a given cdf.
We highlight that while the example uses a discrete random variable, the aggregators are also applicable in the general case; in the same vein, our theory developed in \Cref{sec:Estimates} also applies to general random variables (i.e.\ discrete, continuous, or mixed).
We denote by $X$ the random variable and by $\domain(X)$ its domain; in the example it is (a subset of) the natural numbers.
For $x\in\domain(X)$ we write $p(x)$ for $\mathbb{P}(X=x)$.

\begin{description}
    \item[Mean] Recall that $\mathbb{E}(X) = \int X(\omega) \,d\mathbb{P}(\omega)$.
    As our example comprises a discrete random variable, we simply compute the weighted sum $\mathbb{E}(X) = \sum_{x\in \domain(X)} x \cdot p(x) = 0.5\cdot 1 + \sum_{i=1}^\infty 2i \cdot (\frac 1 2)^{i+1} = 2.5$.
    
    \item[Higher moments] are computed as the mean of $X^n$.
    Essentially, we rescale every outcome by taking its $n$-th power. 
    For example, with $n = 2$ we get $\mathbb{E}(X^2) = \sum_{x\in \domain(X)} x^2 \cdot p(x) = 0.5 \cdot 1^2 + \sum_{i=1}^\infty (2i)^2 \cdot (\frac 1 2)^{i+1} = 12.5$.
    
    \item[Quantiles] Recall that $\Quant_t(X) = \inf \{v \mid \mathbb{P}(X \leq v) \geq t\}$ for $t\in (0,1)$.
    Intuitively, we take the highest value of the worst $\frac t {100}$ \% outcomes.
    We note that quantiles are sometimes also defined by \enquote{partitioning} the outcomes into $n$ different blocks and then taking the $k$-the value, for example \enquote{the second 20\% quantile}, which is equal to $t = 0.4$ in our definition, i.e.\ $t = \frac{k}{n}$ in general.
    
    Choosing $t \in (0,0.5]$, we get $\Quant_t(X) = 1$, since that is the smallest value with a positive probability.
    Then we observe the non-continuity of quantiles that is mentioned as a complication in \cref{sec:ObtainingBounds}: 
    Choosing $t \in (0.5,0.75]$, we obtain $\Quant_t(X) = 2$.
    Quantiles can easily be read off from the cdf by finding the smallest (leftmost) value $x$ such that the cdf is above $t$.
    
    \item[Conditional Value-at-Risk]
    Formally, CVaR is defined as:
    $\CVaR_t(X) = \frac{1}{t} (P \cdot \mathbb{E}[X \mid X < v] + (t - P) \cdot v)$, where $t \in (0,1)$, $v = \Quant_t(X)$ and $P = \mathbb{P}(X < v)$.
    Intuitively, this means taking the expectation over the $t$-quantile.
    The reason this definition looks surprisingly complicated is that it has to account for the probability mass exactly at $v$ potentially only being partially included.
    
    In our example, for $t=0.75$, we have $v=2$ and $P= 0.5$, yielding 
    $\CVaR_{0.75}(X) = \frac{1}{0.75} \cdot (0.5\cdot 1 + 0.25 \cdot 2)$.
    However, choosing $t=0.7$, observe that we only want to consider $0.2$ of the outcome $2$, even though $p(2) = 0.25$, reflected by $(t - P) \cdot v = 0.2 \cdot 2$.
    
    \item[Entropic Risk] Entropic risk is defined as $\ERisk_\gamma(X) = - \frac{1}{\gamma} \log(\mathbb{E}(e^{-\gamma X}))$, with $\gamma > 0$.
    This function first rescales the obtained reward using the exponential function $e^{-\gamma X}$, takes the expectation of the exponentially rescaled rewards, and then re-normalizes the value by $-\gamma \log(\dots)$.
    We refer to~\cite{BCMP24} for a more detailed explanation.

    In our example, for $\gamma = 2$, we get 
    \begin{align*}
    \ERisk_1(X) & = - \tfrac{1}{2} \log({\sum}_{x \in \domain(X)} e^{-2x} \cdot p(x)) \\ 
        & = -\tfrac{1}{2} \log(e^{-2} \cdot 0.5 + {\sum}_{i=1}^\infty e^{-4i} \cdot (\tfrac{1}{2})^{i+1}) \approx 1.35.
    \end{align*}

\end{description}

\section{Sequential DKW}
\label{sec:Sequential}

In the main body, we considered the problem of deriving \enquote{as good as possible} bounds given a sample budget (or an already gathered set of samples).
However, as mentioned in \cref{rem:dual}, sometimes we are also interested in the \emph{sequential} setting, where the goal is to gather samples until a certain precision can be guaranteed.
As a naive approach, we could try to derive an a-priori upper bound on the number of samples required.
For example, \cite{BHMWW25} notes that such bounds for estimating the mean can be derived by using the fact that DKW reduces to Hoeffding's inequality in the worst case.
However, far fewer samples might be sufficient in case the variance of the sample data is low.

Additionally, the \enquote{DKW-Lower} procedure introduced in \cite{BHMWW25} gives a lower bound for the expected value of an unbounded reward, again for a fixed set of samples.
It is shown that DKW-Lower ``converges'' in the sense that the lower bound can be arbitrarily close to the true expected value if a large enough sample set is chosen.

Intuitively, one might think it is possible to derive a simple sequential procedure by repeatedly computing bounds via the DKW inequality until a satisfactory confidence interval is achieved.
However, this simple approach to build sequential procedures is no longer sound:
The probability that \emph{any} confidence interval is incorrect accumulates and is no longer guaranteed to be $\leq\delta$. 
Even if we are only interested in the final confidence interval in the sequence, the guarantees of the DKW inequality no longer apply since the stopping condition is not independent of the outcome of the sampling process, as discussed in \cite{BHMWW25}.

In this section, we now aim to define a sequential procedure that retains soundness while converging to the true expected value.
The core idea is to define certain stages at which we build confidence intervals and split the confidence budget $\delta$ over all stages in such a way that the probability of all stages yielding a correct confidence interval is $\geq 1-\delta$ while at the same time ensuring that the sample count per stage grows fast enough to counteract the growth of the confidence intervals due to decreasing $\delta$.
Further, we show that this sequential procedure extends to a very general class of objectives beyond expected rewards.

Formally, let $\mathcal{F}$ be an aggregation.
(Technically, when we write $\mathcal{F}(F)$ for a cdf $F$, we refer to computing the aggregation of the associated random variable.)
Like in the DKW inequality, for $\varepsilon\geq 0$ we define the confidence band
\[ \hat{F}_\varepsilon=\left\{F'(x)\mid {\sup}_{x\in\RRpos}\lvert F(x)-F'(x)\rvert \leq \varepsilon\right\}.\]
and derive the confidence interval $C_\varepsilon$ as
\[ C_\varepsilon=[\underline{C}_\varepsilon, \overline{C}_\varepsilon]=\left[ {\inf}_{F\in\hat{F}_\varepsilon} \mathcal{F}(F), {\sup}_{F\in\hat{F}_\varepsilon} \mathcal{F}(F)\right]. \]
Given an infinite stream of random variables $\mathcal{X}=X_1,X_2,\dots$ and confidence level $1-\delta$, we define, we define the procedure ``Sequential-DKW'' as follows:
We choose an $n\in\mathbb{N}$ and define $n_i = n i^2$ as well as $\delta_i=\frac{\delta}{2^i}$ for all $i\geq 0$.
We also define $\Xi_i=\{X_1,\dots,X_{n_i}\}$ with the corresponding eCDF $F_i$.
Finally, we define the sequence of confidence intervals output by Sequential-DKW as $(C_i)_{i\geq 0}$ where $C_i=C_{\varepsilon_i}$ with $\varepsilon_i=\sqrt{\frac{\ln(2/\delta_i)}{2n_i}}$.

\begin{theorem}\label{thm:sequential-dkw}
    For any aggregation $\mathcal{F}$, stream of random variables $\mathcal{X}$ drawn from a distribution with cdf $F_X$ and $\delta\in(0,1)$, the probability that all confidence intervals $C_i$ produced by Sequential-DKW are correct is at least $1-\delta$, i.e.
    \[ \mathbb{P}\left( \forall i\geq 0. \mathcal{F}(X) \in C_i \right) \geq 1-\delta. \]
    Further, if $C_\varepsilon$ is always continuous at $\varepsilon=0$, we have 
    \[ {\lim}_{i\rightarrow\infty} \overline{C}_i = {\lim}_{i\rightarrow\infty}\underline{C}_i = \mathcal{F}(F_X).\]
\end{theorem}

\begin{proof}
    By the DKW inequality, he have 
    \begin{align*}
        \mathbb{P}\left( \mathcal{F}(F_X) \notin C_i \right) & \leq \mathbb{P}\left( F_X \notin \hat{F}_{\varepsilon_i} \right) \\
        & = \mathbb{P}\left( {\sup}_{x\in\RRpos} |F_X - F_i| \geq \sqrt{\ln(2/\delta_i) / (2n_i)} \right) \\
        & \leq 2e^{-2n_i\sqrt{\ln(2/\delta_i) / (2n_i)}^2} \\
        & = \delta_i
    \end{align*}
    By definition of $\delta_i$ and the union bound we have
    \begin{equation*}
        \mathbb{P}\left( \forall i\geq 0. \mathcal{F}(F_X) \in C_i \right) \leq {\sum}_{i=0}^\infty\mathbb{P}\left( \mathcal{F}(F_X) \notin C_i \right) \leq {\sum}_{i=0}^\infty \delta_i \leq {\sum}_{i=0}^\infty \frac{\delta}{2^i} = \delta
    \end{equation*}
    To show the second part of the theorem, note that for all $i\geq 0$, by definition of $C_{\varepsilon_i}$, there is a sequence $(\underline{F}_i^j)_{j\geq 0}$ for which all $\underline{F}_i^j\in\hat{F}_{\varepsilon_i}$ and $\lim_{j\rightarrow\infty}\mathcal{F}(\underline{F}_i^j)=\underline{C}_i$ and and analogous sequence $(\underline{F}_i^j)_{j\geq 0}$.
    Then, since $\underline{F}_i^j,\overline{F}_i^j\in\hat{F}_{\varepsilon_i}$ for all $j\geq 0$ and by definition of $n_i$ we have
    \begin{multline*}
        {\sup}_{x\in\mathbb{R}}\left({\lim}_{j\rightarrow\infty}\overline{F}_i^j(x)-{\lim}_{j\rightarrow\infty}\underline{F}_i^j(x)\right) \leq 2\sqrt{\ln(2/\delta_i) / (2n_i)} \\
        = 2\sqrt{\ln(2^{i+1}/\delta) / (2ni^2)}
        = 2\sqrt{(i+1)\ln(2/\delta) / (2ni^2)}
    \end{multline*}
    Thus, since $\lim_{i\rightarrow \infty}\varepsilon_i=0$ and $C_\varepsilon$ is continuous at $\varepsilon=0$, we have 
    \[
        {\lim}_{i\rightarrow\infty}\overline{C}_i-\underline{C}_i = {\lim}_{i\rightarrow\infty}\left({\lim}_{j\rightarrow\infty}\mathcal{F}(\overline{F}_i^j)-{\lim}_{j\rightarrow\infty}\mathcal{F}(\underline{F}_i^j)\right)=0,
    \]
    i.e.\ $\lim_{i\rightarrow\infty}\overline{C}_i=\lim_{i\rightarrow\infty}\underline{C}_i$.
    Finally, by the law of large numbers $\lim_{i\rightarrow\infty}F_i=F_X$ and thus again by $\lim_{i\rightarrow \infty}\varepsilon_i=0$ and continuity of $C_\varepsilon$ at $\varepsilon=0$ we have $\lim_{i\rightarrow\infty}\overline{C}_i=\lim_{i\rightarrow\infty}\underline{C}_i=\mathcal{F}(F_X)$.
\end{proof}

\begin{remark}
    Sequential-DKW naturally extends to bounded random variables (e.g.\ positive expected reward with an a priori upper bound) where a $\leq X \leq b$ for all $X\in\mathcal{X}$ by additionally requiring $F(a)=0$ and $F(b)=1$ for all $F\in\hat{F}_{\varepsilon_i}$ for all $i\geq 0$.
\end{remark} %

\begin{remark}
    If only $\underline{C}_i$ is continuous at $\varepsilon=0$ we can still obtain convergence for the lower bound, i.e.\ $\lim_{i\rightarrow\infty}\underline{C}_i=\mathcal{F}(F_X)$, in the same way as for \Cref{thm:sequential-dkw}.
    The analogous statement holds for $\overline{C}_i$.
    This is for example relevant for unbounded expected rewards (where $\overline{C}_\varepsilon$ is always finite for $\varepsilon=0$ but infinite for $\varepsilon>0$).
    There, $(\underline{C}_i)_{i\geq0}$ as in Sequential-DKW is a sound sequence of lower bounds converging towards the true expected reward, but Sequential-DKW likely does not yield converging upper bounds since $\overline{C}_i=\infty$ for all $i\geq 0$.
    Similarly, for VaR $\overline{C}_\varepsilon$ may not always be continuous at $\varepsilon=0$ for discrete distributions but $\underline{C}_i$ is.
\end{remark}

}
{
}

\end{document}